\documentclass[12pt]{amsart}
\usepackage{a4wide}
\usepackage{amssymb}
\usepackage{amsthm}
\usepackage{amsmath}
\usepackage{amscd}
\usepackage[foot]{amsaddr}
\usepackage{verbatim}
\usepackage{color}
\usepackage[mathscr]{eucal}
\newtheorem{theorem}{Theorem}
\newtheorem{lemma}[theorem]{Lemma}

\newtheorem{proposition}[theorem]{Proposition}
\newtheorem{definition}[theorem]{Definition}

\theoremstyle{remark}

\numberwithin{theorem}{section} \numberwithin{equation}{section}

\newcommand{\R}{\mathbb{R}}
\newcommand{\C}{\mathbb{C}}
\newcommand{\Tr}{{\text {\rm Tr}}}

\newcommand{\Z}{\mathbb{Z}}
\newcommand{\N}{\mathbb{N}}

\newcommand{\bea}{\begin{eqnarray}}
\newcommand{\eea}{\end{eqnarray}}
\newcommand{\be}{\begin{equation}}
\newcommand{\ee}{\end{equation}}
\newcommand{\ben}{\begin{eqnarray*}}
\newcommand{\een}{\end{eqnarray*}}
\newcommand{\bem}{\begin{pmatrix}}
\newcommand{\eem}{\end{pmatrix}}
\newcommand{\bal}{\begin{align}}
\newcommand{\eal}{\end{align}}

\newcommand{\CN}{{\mathcal N}}

\newcommand{\CH}{{\mathcal H}}
\newcommand{\ndt}{\noindent}

\newcommand{\wt}{\widetilde}
\newcommand{\IZ}{{\mathbb Z}}

\newcommand{\D}{{\Delta}}
\newcommand{\vth}{{\vartheta}}
\newcommand{\myh}{{\xi}}

\def\t{\tau}
\def\s{\sigma}
\def\v{\varphi}
\def\l{\lambda}
\def\m{\mu}

\def\={\;  = \;}
\def\+{\, + \,}
\def\inn{\,\in\,}

\def\mypmod#1{\, {(\rm mod} \, {#1})}

\begin{document}

\hspace{13cm} \mbox{\small NIKHEF2012-017}

\vspace{1cm}

\title{On the positivity of black hole degeneracies in string theory}
\author{Kathrin Bringmann$^1$}
\address{$^1$Mathematical Institute, University of Cologne\\ Weyertal 86-90, 50931 Cologne, Germany}
\email{kbringma@math.uni-koeln.de}
\author{Sameer Murthy$^2$}
\email{murthy.sameer@gmail.com}
\address{$^2$NIKHEF theory group, Science Park 105 \\ 1098 XG Amsterdam, The Netherlands}
\begin{abstract}
Certain helicity trace indices of charged states in $\CN=4$ and $\CN=8$ superstring theory have been
computed exactly using their explicit weakly coupled microscopic description.
These indices are expected to count the exact quantum degeneracies of black holes
carrying the same charges.
In order for this interpretation to be consistent, these indices should be positive integers.
We prove this positivity property for a class of four/five dimensional black holes
in type II string theory compactified on $T^{6}/T^{5}$ and on $K3 \times T^{2}/S^{1}$.
The proof relies on the mock modular properties of the corresponding generating functions.
\end{abstract}
\maketitle

\section{Introduction and statement of results}
The study of supersymmetric black holes in string theory has been very effective in
shedding light on the issue of black hole entropy.
The strength of the string theoretic approach lies in the fact that there
are two related descriptions of charged black holes. The first (strong string coupling, macroscopic) description
is a low energy effective description as general relativity coupled to a set of matter fields.
In the second (weak string coupling, microscopic) description,  a generic
state of the theory with the same charges as the black hole is identified as a collection of
fundamental objects of string theory namely fundamental strings and branes.
The fluctuations of these objects make up the elementary excitations
(microstates) of the theory, which can be described by conventional quantum field theoretic methods.
The key idea is to identify these two descriptions valid at strong and weak coupling, respectively.
At strong coupling, the excitations of the strings and branes exert a gravitational force on each
other, and the black hole can be thought of as a quantum mechanical bound state of these microstates.

In a class of supersymmetric string theories with sixteen or more unbroken supercharges
we now have a practically complete understanding of the spectrum of BPS states (see~\cite{Mandal:2010cj}
for a relatively recent review).
One can therefore subject the above idea to high precision tests, by comparing
the statistical entropy of the ensemble of states and an appropriately defined thermodynamic entropy of
the corresponding BPS black hole, beyond a large charge approximation. Since we know the
microscopic degeneracies exactly, one can even aim for an exact comparison for
finite charges using the framework of the quantum entropy formalism \cite{Sen:2008yk, Sen:2008vm}.
The first such comparisons have been successfully
performed in highly supersymmetric examples \cite{Dabholkar:2010uh, Dabholkar:2011ec},
and this has been expressed as an exact (finite $N$) $AdS_{2}/CFT_{1}$ correspondence:
\be\label{exacthol}
d_{\rm hor}(n_{i}) = d_{\rm micro} (n_{i}),
\ee
where $\{ n_{i} \}$ are the quantized charges of the black hole, $d_{\rm hor}$ is the macroscopic
black hole entropy, and $d_{\rm micro}$ is the microscopic degeneracy.

In carrying out such a comparison there is an important subtlety. On the macroscopic side,
the black hole entropy is supposed to calculate the logarithm of the absolute degeneracy
of states $d_{\rm hor}$ according to the Boltzmann relation. On the other hand, on the microscopic side,
one normally computes a supersymmetric index (like a helicity supertrace),
and so $d_{\rm micro}$ is a difference of the number of bosonic and fermionic multiplets.
These two quantities are a priori not the same, but as we  review below, it has been
argued that holography gives an explanation of their equality \cite{Dabholkar:2010rm, Sen:2009vz}.

Due to the interpretation of $d_{\rm hor}$ as computing the number of states of the black hole,
an immediate consequence of the equality \eqref{exacthol} is that the microscopic index
$d_{\rm micro}$ should be a positive integer. The known examples of exact BPS
(indexed) counting formulas are all related to Fourier coefficients of automorphic
forms of various types, and for the explicit automorphic forms under consideration
(discussed below), it is not at all obvious that their Fourier coefficients obey this positivity criterion.
Positivity can thus be thought of as a prediction from the quantum theory of black holes\footnote{Strictly
speaking, one has a notion of a black hole in the gravitational theory only in a large charge approximation.
At infinite charges, one is in a classical two derivative theory with a well-defined notion of a horizon.
The $1/n_{i}$ corrections can be thought of as slightly changing the location of the horizon in spacetime.
However, at small values of the charges, spacetime is highly curved and it is possible that classical
notions completely break down. The prediction assumes that there is still some sense in semi-classical
reasoning for all charges.}
for Fourier coefficients of automorphic forms.
Checking the prediction for positivity is therefore a (perhaps coarse, but) important check of our
understanding of black holes in quantum gravity.
Our aim in this paper is to prove the positivity criterion for a class of black holes in theories
with $\CN \ge 4$ supersymmetry.

\vspace{0.2cm}

\ndt \emph{Index = Degeneracy}

The supersymmetric index receives contribution only from BPS states and hence is protected from
any change under continuous deformations of the moduli of the theory, so (in the absence
of wall-crossings) the microscopic index is the same as the macroscopic index.
In the macroscopic theory, the index can be argued to be equal to the degeneracy as follows.
The near-horizon geometry of the supersymmetric black hole always has an
$AdS_{2}$ factor, which has an $SU(1, 1)$ symmetry. If the black hole geometry leaves at least four supersymmetries unbroken, then the closure of the supersymmetry algebra requires that the near horizon
symmetry must contain the supergroup $SU(1, 1|2)$, the bosonic $SU(2)$ R-symmetry being identified
with spatial rotations. This means that the horizon states on an average have zero
charge under the Cartan generator $J$ of this $SU(2)$. The $AdS_{2}$ geometry further fixes the
theory to be in the microcanonical ensemble, which implies that, in fact, every state in the ensemble
has~$J=0$.  So we have
\be\label{ind-deg}
\Tr (-1)^{J} = \Tr (1) \, ,
\ee
that is, index equals degeneracy.
For a more detailed discussion see \cite{Dabholkar:2010rm}.

Note the the index equals degeneracy only for the horizon degrees of freedom, but usually one
does not compute the index of the horizon degrees of freedom directly. It is easier to compute
the index  of the asymptotic states as a spacetime helicity supertrace which receives contribution
also from the degrees of freedom external to the horizon. It is crucial that the  contribution of
these external modes is removed from the helicity supertrace before checking the equality \eqref{ind-deg}.
Typically, modes localized outside the horizon come from three sources \cite{Sen:2010mz}
-- fluctuations of supergravity fields around the black hole solution, non-linear gravitational
configurations like multi-centered black holes, and fermion zero modes.

The field fluctuations localized outside the horizon come from fields that carry NS-NS charges such
as the momentum, but not from those that carry D-brane charges \cite{Banerjee:2009uk,Jatkar:2009yd}.
In a duality frame where all charges come from D-branes, one therefore does not have to worry
about these external field fluctuations.
The contributions of multi-centered black holes, when present, have to be explicitly subtracted.
For theories that preserve 16 or more supercharges, black hole solutions with three or more centers
are expected not to contribute to the index \cite{Dabholkar:2009dq}. In these situations, one still has to
subtract the contribution from two-centered configurations, as we do in explicit examples in this paper.
The third source, i.e.~fermion zero modes are generically present and one has to
deal with them explicitly. We assume that the only fermion zero modes present are those arising
from broken supersymmetry. In that case, they contribute an overall (positive or negative) rational
constant to the index, which one has to factor out as is done in \cite{Dabholkar:2010rm, Sen:2010mz}.

\smallskip

\ndt {\it The positivity conjecture}
\smallskip

Putting together the above discussion, one can make the following precise conjecture about the sign of the
index of BPS states in any given string theory:
$d_{\rm micro}(n_{i})>0$ whenever a black hole solution carrying the corresponding charges $(n_{i})$ can exist.
This positivity conjecture was presented by Ashoke Sen at the ASICTP school on  modular
forms and their applications in March 2011.
In this paper, we prove this conjecture for a class of black hole in~$\CN=4$ and~$\CN=8$ string theory.

We now make a brief list of the various black holes that we study in this paper, along with the corresponding
automorphic form that controls their degeneracies. The formulas for the microscopic degeneracies~$d_{\rm micro}$
as a function of the black hole charges in each case will be given in the bulk of the paper.
We will study four and five dimensional string theories with 32 supercharges
(case 1) and 16 supercharges (case 2). The four dimensional black holes are:
\begin{enumerate}
\item[1a.] \emph{1/8-BPS black holes in type II string theory on $T^{6}$}.
These black holes are labelled by an integer $\Delta$, and $d^{(1a)}_{\rm micro}$
is given in terms of the Fourier coefficients $c(n,r)$ with $\Delta = 4n-r^{2}$ of
$\v_{-2,1}(\tau, z)$, a weight $k=-2$ and index $m=1$ weak Jacobi form.
\item[2a.] \emph{1/4-BPS  black holes in  type II string theory on $K3 \times T^{2}$}.
Here, the black holes are labelled by three integers $(n,r,m)$, and
$d^{(2a)}_{\rm micro}$ is given in terms of the Fourier coefficients of the Siegel
modular form $1/\Phi_{10}(\s,\t,z)$ expanded in the ``attractor region''.
\end{enumerate}
In the corresponding five-dimensional situations (1b, 2b),
the theories are related to their  four-dimensional counterparts
by a decompactification of one of the circles of the~$T^{2}$, and the
black holes in these theories are related to their four-dimensional counterparts by the
4d-5d lift~\cite{Gaiotto:2005gf}. %

For the automorphic forms written above, the existence of the black hole solution implies
that the discriminant~$4mn-r^{2}>0$.
Case (1a) is very simple to prove, the proof is simply a statement of the positivity of the
Fourier coefficients of the canonical Jacobi theta function and negative powers of the eta function,
as was already mentioned
in \cite{Dabholkar:2011ec}. Case (1b) follows with little work, we present the proof below.
Case (2a) and (2b) are more difficult to prove, the main reason being that the
expansion of the Siegel form in the attractor region destroys the automorphic
properties of the generating function.
Numerical evidence for case (2a) was first written in \cite{Sen:2010mz}.
We can perform a Fourier expansion in the $\tau'$ variable without a problem and
the Fourier coefficient of $e^{2 \pi \i m \tau'}$ is a Jacobi form of weight $-10$
and index $m$. This Jacobi form is meromorphic in the $z$ variable, and one therefore needs
to specify the contour to define its Fourier expansion, and the Fourier expansion breaks
the modular properties. However, it has been shown that $d_{\rm micro}^{(2a)}$ is a Fourier
coefficient of a mock Jacobi form \cite{Dabholkar:2011mm},
and one can recover a remnant of the modular properties in a very elegant way.
Although the full theory for these objects is not known, enough is known to do a case-by-case
analysis in the magnetic charge invariant~$m$.

In this paper, we prove that $d_{\rm micro}>0$ in the cases (2a), (2b) for $m=1,2$ for all values
of~$(n,r)$ with $4mn-r^{2}>0$.
We present two proofs of the positivity in the case $m=2$.
The first proof uses the Circle Method.  In the case of modular forms this method only requires knowing
the weight and the principal part of the modular form in all cusps.   This was extended by the first author
and Ono  \cite{BO1, BO2} to  mock modular
forms and by the first author and Mahlburg to mixed mock modular forms \cite{BM}.
 The second method is complementary in that we use
the explicit knowledge of the full functions, but we can write down an algebraic proof that holds
for all coefficients. It relies on the explicit knowledge of the
modular and mock modular forms in our examples, and simple algebraic facts about
the basic building blocks of modular forms -- theta functions,
eta functions, and the Eisenstein series -- and a simple estimate for the Hurwitz-Kronecker
class numbers. Both these methods need us to specify the value of the index $m$.
Although we work out the first two cases $m=1,2$ here, both our proofs can be extended to higher
values of $m$ case-by-case. It would be nice, however, to come up with a proof which
tackles all values of $m$ at one shot.

The remainder  of the paper is structured  as follows.
In \S\ref{1/8}, we use the concrete set up of the string theory on $T^{6}$  (Case~1)
to briefly review Jacobi forms and some of their properties useful for our application.
Using the same set up, we then discuss the lift to five dimensions, and the different
ensembles of rotating black holes. In \S\ref{1/44d}, we address the theories on $K3\times T^{2}$
(Case~2), and analyze the explicit mock Jacobi forms which arise for index $m=1,2$.
In~\S\ref{mone}, we prove the positivity property for~$m=1$, and in \S\ref{mtwo}, and \S\ref{alter},
we prove the positivity property for~$m=2$ in two different ways.
In an appendix, we give some tables listing the first few values of the black hole
degeneracies in Case~2 for~$m=1,2,3,4$.

\section*{Acknowledgements}
The research of the first author was supported by the Alfried Krupp Prize for Young University
Teachers of the Krupp Foundation. The research of the second author was supported by
the ERC Advanced Grant no. 246974, {\it ``Supersymmetry: a window to non-perturbative physics''}.
We thank the ASICTP, Trieste for hospitality during the 2011 conference on `Modular Forms and mock
modular forms, and their applications in arithmetic, geometry, and physics' where this research was initiated.
It is a pleasure to thank Atish Dabholkar, Ashoke Sen, and Don Zagier for useful discussions.

\section{Black hole degeneracies and Jacobi forms \label{1/8}}

\subsection{Review of Jacobi forms}

The black hole microstate degeneracies in all the cases that we study are related to Fourier
coefficients of Jacobi (or mock Jacobi) forms. We therefore begin by recalling
a few relevant facts about Jacobi forms \cite{Eichler:1985ja}.
We use the notation $e(x) :=e^{2 \pi \i x}$,
$q:=e(\t)$, and $\zeta := e(z)$, which is fairly standard in the modular forms literature.

\vspace{0.2cm}

\begin{definition}
A \emph{Jacobi} form of  weight $k\in \Z$ and index $m\in \Z$  is a  holomorphic function $\varphi:\mathbb{H} \times\C \to \C$ which
is ``modular in $\tau$ and elliptic in $z $'' in the sense that it transforms under the modular group as
  \be\label{modtransform}  \varphi \left(\frac{a\t+b}{c\t+d}, \frac{z}{c\t+d}\right) \ =
   (c\t+d)^k \, e\left({\frac{m c z^2}{c\t +d}}\right) \, \varphi(\t,z)  \qquad  \quad
\left(  \forall \left(\begin{smallmatrix} a & b \\ c & d   \end{smallmatrix} \right)    \in \rm{SL}_2(\mathbb{Z}) \right)
\, ,  \ee
and under the translations of $z$ by $\mathbb{Z} \tau + \mathbb{Z}$ as
  \be\label{elliptic}  \varphi(\t, z+\lambda\tau+\mu)\= e\left({-m\left(\lambda^2 \t + 2 \lambda z\right)}\right) \varphi(\t, z)
  \qquad \left( \forall  \l,\,\m \in \mathbb{Z} \right).
  \ee
\end{definition}

\vspace{0.2cm}

\ndt \emph{Fourier expansion:}
Equations \eqref{modtransform} and (\ref{elliptic}) include the periodicities $\varphi(\t+1,z) = \varphi(\t,z)$ and $\varphi(\t,z+1) = \varphi(\t,z)$, thus  $\varphi$ has a Fourier expansion
  \be\label{fourierjacobi}  \nonumber
  \varphi(\t,z) \= \sum_{n, r} c(n, r)\,q^n\,\zeta^r.
   \ee
Equation \eqref{elliptic} is then equivalent to the periodicity property
  \be\label{cnrprop}
   c(n, r) \= C\left(4 n m - r^2, r\right),
\text{where $C(\D,r)$ depends only on $r \mypmod{2m}$.}
  \ee

The function $\v(\tau, z)$ is called a  \emph{holomorphic Jacobi form} (or simply a \emph{Jacobi form})
of weight $k$ and index $m$ if if it satisfies the condition
  \be\label{holjacobi}  c(n, r) \= 0 \qquad \textrm{unless} \qquad 4mn \ge r^2\,. \ee
The function  is called a \emph{weak Jacobi form} if it satisfies the condition
  \be\label{weakjacobi} c(n, r) \= 0\qquad   \textrm{unless}  \qquad n \geq 0 \, .\ee
The Jacobi forms that arise as the generating functions of black hole degeneracies
are always weak which is related to the fact that the condition \eqref{weakjacobi}
is equivalent to an exponential growth of $C(\D,r)$ as $\D \to \infty$.

\vspace{0.2cm}

\ndt \emph{Theta expansion:}
Using  the transformation property (\ref{elliptic})  one obtains that the
Fourier expansion of a Jacobi  form may be written as
  \be\label{jacobi-Fourier}
  \v(\t, z) \= \sum_{\ell\inn \IZ} \;q^{\frac{\ell^2}{4m}}\;h_\ell(\t) \; e(\ell z) \, ,
   \ee
where $h_\ell(\tau)$ is periodic in $\ell$ with period $2m$.  In terms of the coefficients \eqref{cnrprop} we have for $\ell \inn \IZ/2m \IZ$
  \be\label{defhltau} \nonumber
   h_{\ell}(\t) \= \sum_{\D} C_{\ell}(\D) \,  q^{\frac{\D}{4m}} \, .
    \ee
Because of the periodicity property of $h_{\ell}$, equation \eqref{jacobi-Fourier} can be rewritten in the form
  \be\label{jacobi-theta}
  \v(\t,z) = \sum_{\ell\inn \IZ/2m\IZ} h_\ell(\t) \, \vartheta_{m,\ell}(\t, z)\,, \ee
where $\vartheta_{m,\ell}(\t,z)$ denotes the standard index $m$ theta function
  \be \label{thetadef}
   \vartheta_{m,\ell}(\t, z)
   \;:=\; \sum_{{\l\inn\IZ} \atop {\l\,\equiv\,\ell\,\mypmod{2m}}} q^{\frac{\l^2}{4m}} \, \zeta^\l \, \,
   \= \sum_{n \inn \mathbb{Z}} \,q^{m\left(n+ \frac{\ell}{2m}\right)^2} \,  \zeta^{\ell + 2mn} \, .
    \ee
The expansion~\eqref{jacobi-theta} is called  the \emph{theta expansion} of $\v$.
The vector $h := ( h_1, \ldots, h_{2m})$ transforms like a modular form of weight $k-\frac{1}{2}$ under SL$_2(\IZ)$ with respect to the Weyl representation.

\vspace{0.2cm}

\ndt \emph{Hecke-like operators:}
We will require   the Hecke-like operator $V_{t}$ ($t \ge 1$),
which sends Jacobi forms of weight $k$ and index $m$ to Jacobi forms of weight $k$ and index $t m$.
It is given in terms of its action on Fourier coefficients by
  \be\label{defVl} V_t\,: \quad \sum_{n,r} c(n,r)\,q^n\,\zeta^r \quad \mapsto \quad
    \sum_{n,r} \left( \sum_{d|(n,r,t)} d^{k-1} c\left(\frac{nt}{d^2},\frac{r}{d}\right) \right) q^n\,\zeta^r \ . \ee

\vspace{0.2cm}

\ndt \emph{Jacobi forms of index one:}
  If $m=1$, \eqref{cnrprop} reduces to $c(n,r) = C\left(4n-r^2\right)$,
   where $C(\D)$ is a function of a single argument.
Two examples of index~1 Jacobi forms, which play an important role in the theory, are the following two
Jacobi forms of weight $-2$ and $0$, respectively:
\be\label{phi2}
A(\tau, z)=\v_{-2,1}(\tau, z) :=  \frac{\vth_1^2(\t, z)}{\eta^6(\t)} \, ,
\ee
 \be\label{phi0}
 B(\tau, z)=\v_{0, 1} (\t, z) := 4 \left( \frac{\vth_2^2(\t, z)}{\vth_2^2(\t)} +
    \frac{\vth_3^2(\t, z)}{\vth_3^2(\t)} +\frac{\vth_4^2(\t, z)}{\vth_4^2(\t)} \right) \, ,
\ee
where $\vartheta_{i}, i=1,\dots,4$ are the four classical Jacobi theta functions
\begin{align}
 \nonumber
 & \vartheta_1(\tau, z) :=\sum_{n\in\Z}(-1)^nq^{\frac12\left(n-\frac12\right)^2}\zeta^{n-\frac12}, \, \qquad
\vartheta_2(\tau, z) :=\sum_{n\in\Z}q^{\frac12\left(n-\frac12\right)^2}\zeta^{n-\frac12},\\
 \nonumber
 & \qquad \vartheta_3(\tau, z) :=\sum_{n\in\Z}q^{\frac{n^2}{2}}\zeta^n, \,  \qquad \qquad \qquad
\vartheta_4(\tau, z) :=\sum_{n\in\Z}(-1)^nq^{\frac{n^2}{2}}\zeta^n \,\nonumber ,
\end{align}
and~$\eta$ is the Dedekind eta function
\be
\eta(\tau) := q^{\frac{1}{24}} \prod_{n\geq 1}(1 - q^{n}) \, .  \label{defeta}
\ee

By the property mentioned above, these functions have a Fourier expansion ($k=-2,0$):
\be\label{phiC} \v_{k, 1}(\tau, z) \= \sum_{n,\,r\in\Z} C_k(4n -r^2)\,q^n\,\zeta^r.
 \ee
The first few Fourier coefficients of $A$ and $B$ are given in Table \ref{tablefcoeffs} below.
Note the alternating sign pattern of~$C_{k}(\D)$. This is related to the positivity of the black hole
degeneracies, and we will prove that this is true for all~$\D>0$.

\begin{table}[h]   \caption{\small{The first few Fourier coefficients of $A$ and $B$}}   \vspace{5pt}    \centering
   \begin{tabular}{|c|ccccccccc|ccccc}   \hline
        $k$ &  $C_{k}(-1)$ & $C_{k}(0)$ &$C_{k}(3) $& $C_{k}(4)$& $C_{k}(7)$ & $C_{k}(8)$ &$C_{k}(11)$& $C_{k}(12)$& $C_{k}(15)$ \\
       \hline $-2$ & 1 & $-2$ &8 &$-12$&39&$-56$&152&$-208$&513\\   $0$ &  1  & 10 &$-64$&108&$-513$&808&$-2752$&4016&$-11775$\\
     \hline    \end{tabular}   \label{tablefcoeffs}   \end{table}

It is a fact that $A$ and $B$ generate the ring of weak Jacobi forms of even weight
freely over the ring of modular forms of level~$1$  \cite{Eichler:1985ja}, which means that any weak Jacobi form can be
written as a sum of products of $A$ and $B$ with coefficients being modular forms.

\subsection{1/8 BPS black holes in type II string theory on $T^{6}$}

On compactifying type-II string on a 6-torus $T^{6}$,
the resulting four-dimensional theory has $\CN=8$ supersymmetry with $28$ massless $U(1)$ gauge fields.
A charged state is therefore characterized by $28$ electric and $28$ magnetic charges which combine
into the $\bf 56$ representation of the U-duality group $E_{7, 7} (\mathbb{Z})$.
We are interested in one-eighth BPS dyonic states in this theory which perserve four of the thirty-two supersymmetries.

We consider the 6-torus to be the product  $T^4 \times S^{1} \times \wt S^{1}$ of a 4-torus and
two circles. Using the U-duality, we can go to a frame where the four dimensional system contains
$Q_5$ D5-branes along $T^4\times S^1$, $Q_1$ D1-branes along $S^1$, and K Kaluza-Klein monopoles
associated with $\wt S^1$, carrying $n$ units of momentum along $S^1$ and $J$ units of momentum
along $\wt S^1$. The black holes are thus labeled by these five charges $(Q_{1}, Q_{5}, n, K, J)$.

We restrict our analysis here to the case $\gcd\big(Kn, Q_1Q_5, KQ_1, KQ_5, nQ_1, nQ_5\big)=1$.
The degeneracies of the 1/8-BPS dyonic states in the type II string theory on a $T^{6}$
are given in terms of the Fourier coefficients of $A(\tau, z)$
\cite{Maldacena:1999bp, Pioline:2005, Shih:2005qf, Sen:2008ta}:
\be \label{et72} \nonumber
d^{(1b)}_{\rm micro} (Q_{1}, Q_{5}, K, n,J) =  (-1)^{J+1}
\sum_{s|Q_{1}, nQ_{5}, J} s \, C_{-2} \left(\left(4Q_1 Q_5 K n -J^{2}\right)/s^{2}\right)\, ,
\ee
where $C_{-2}(D)$ is defined in equation \eqref{phiC}.
The factor of $(-1)^{J+1}$ arises due to the fermion zero modes mentioned in the introduction, which
one has to strip off since they live outside the horizon.

To read off $C_{-2}(D)$ more systematically, we use the  theta expansion
\be\label{jacobi-theta2} \nonumber
A(\t,z)  = h_0(\t) \, \vartheta_{1,0}(\t, z)\, +  h_1(\t) \, \vartheta_{1,1}(\t, z)\,  .
\ee
The functions $h_{\ell}(\tau)$ in this case are given explicitly by:
\bea \label{h0h1defs1}
\nonumber h_{0} (\t) & = & - \frac{\vth_{1,1}(\t,0)}{\eta^{6}(\t)} = -2  -12 q - 56 q^{2}- 208 q^{3} - 684q^4 - 2032q^5 -O\left(q^{6}\right) ,\\
\label{h0h1defs2}
\nonumber h_{1} (\t) & = & \frac{\vth_{1,0}(\t,0)}{\eta^{6}(\t)} =   q^{-\frac{1}{4}}
\left(1 + 8 q + 39 q^{2} + 152q^3 + 513q^4 + 1560q^5 + O\left(q^{6}\right) \right) \, .
\eea
From the definition \eqref{thetadef} of the functions $\vth_{m,\ell}$, and the product
representation~\eqref{defeta} of the function~$\eta$,
it is clear that the Fourier coefficients of $-h_{0}$ and $h_{1}$
are all positive, thus proving the positivity of $d_{\rm micro}$ in this case (1a).

\subsection{Lift to five dimensions, and ensembles with varying and fixed $J_{R}$ \label{5dlift}}

In the above charge representation, zooming in on the tip of the KK monopole gives us
the five dimensional theory. These five dimensional black holes
are therefore labelled by four integers $(Q_{1}, Q_{5}, n, J)$. Near the tip of
the monopole, the circle $\wt S^{1}$ has decompactified, and $J$ becomes an angular momentum charge.
To compute the generating function for the index, one has to remove the modes which are outside
the horizon of the black hole. The only such modes in this case are the bound states of angular
momentum, removing them gives the generating function for $d_{\rm micro}$ \cite{Sen:2011cj}:
\be \label{ei6full}
\nonumber
\sum_J (-1)^{J+1}\, d^{(1a)}_{\rm micro}(n,Q_1,Q_{5},J) \, \zeta^{J} =   \zeta^{-2} \left(\zeta - 1\right)^4 \sum_{j\in \IZ}\sum_{s|n,Q_1 Q_{5},j}
s \, C_{-2} \left({4 Q_1 Q_{5} n - j^2\over s^2}\right) \zeta^{j} \, .
\ee
We have already seen that $(-1)^{d+1}C_{-2}(d)>0$.
The prefactor
$$
\zeta^{-2}(\zeta-1)^{4} = \sum_{r} c_{\rm pf}(r) \, \zeta^{r}
$$
 also has the positivity property
\be
\nonumber
(-1)^{r} c_{\rm pf}(r) >0 \, .
\ee
Putting these two facts together, we get $d^{(1b)}_{\rm micro} >0$.

So far, we have been working with superconformal indices with fixed values for all the charges
including the angular momentum $J$, but in computing the index, one lets $J^{2}$ vary.
In~\cite{Sen:2011cj}, Sen also defined a new index for rotating black holes with fixed value of $J$,
as well as fixed $J^{2}$ (and fixed value of all other charges) as:
\be
\nonumber
d^{\rm rot}_{\rm micro} (\dots,J) := d_{\rm micro} (\dots,J) - d_{\rm micro} (\dots,J+2) \, ,
\ee
where the $\dots$ indicate all the other charges that are held fixed,
and conjectured that this should also be a positive integer.
To prove this,  we need to show that $d^{(1a)}(n,Q_1,Q_{5},J) > d^{(1a)}(n,Q_1,Q_{5},J+2)$.
For the case $(Q_{1}, Q_{5}n) =1$, we need to show that $|C_{-2}(D)|$ are monotonic,
which can be seen from the fact that the function~$\vth_{m,\ell}$ has coefficients one, and
the Fourier coefficients of the function~$\eta^{-6}$, which count partitions, are monotonic.

\section{1/4 BPS black holes in type II string theory on $K3 \times T^{2}$ \label{1/44d}}

The four-dimensional theory in $\R^{1,3}$ resulting from the $K3 \times T^{2}$ compactification has $\mathcal{N}=4$
supersymmetry.
The bosonic duality group of the theory is $SL_{2}(\Z) \times O(22, 6, \mathbb{Z})$, the two factors
are called the $S$-duality group, and the $T$-duality group, respectively.
The integral electric and magnetic charges $(N^{i}, M^{i})$, $(i=1,2,\dots 28)$, are
in a $(2,28)$ representation of this group, and the
degeneracies are written in terms of the T-duality invariants $(N^{2}/2, N\cdot M, M^{2}/2, ) \equiv (n,\ell,m)$,
formed using a certain inner product on the lattice of charges.
The degeneracy formula was first conjectured in~\cite{Dijkgraaf:1996it}, and the complete degeneracy
formula was derived in~\cite{Gaiotto:2005gf, Shih:2005uc, David:2006yn}.
As in the previous subsection, we shall restrict our attention here to the case to primitive charges,
the corresponding formulas for non-primitive charges \cite{Banerjee:2008ri, Banerjee:2008pu, Dabholkar:2008zy}
are related to the primitive degeneracies.

The main novelty (and difficulty) in this case arises because the 1/4-BPS spectrum of the theory depends
not only on the charges, but also the moduli fields at infinity. For given charges $(Q^{i},P^{i})$,
one has, at a generic point in moduli space, not only the dyonic black hole solution,
but also two-centered black hole bound state solutions with the two centers carrying e.g.
electric and magnetic charges \cite{Denef:2000nb}.
These bound states exist only inside a certain region of moduli space bounded by codimension
one surfaces called walls, and cease to exist (decay) on crossing these walls.

On the microscopic side, the (indexed) degeneracies of the 1/4-BPS states are Fourier coefficients of the
meromorphic Siegel modular form $\Phi_{10}^{-1}$, the reciprocal of the Igusa cusp form of weight 10.
The meromorphicity means that the Fourier coefficients depend on the order of expansion, or, in other
words, on the contour of integration one uses to define them. This contour depends on the moduli fields
of the theory, in such a way that the jumps in the degeneracies across the divisors of $\Phi_{10}$
are exactly equal to the degeneracies of the two-centered black hole bound state that decays on
crossing the corresponding wall in moduli space \cite{Cheng:2007ch, Dabholkar:2007vk, Sen:2008ht}.

Our interest is in the degeneracies of the single-centered black hole, and we would like to throw
away the contribution from all the multi-centered black holes to the generating function. This
latter contribution is not modular invariant by itself, and so this step breaks the modular invariance of
the original generating function. However, quite remarkably, the remaining function that one gets has the
property of being a mock Jacobi form \cite{Dabholkar:2011mm},
and this is what we use to prove the positivity of the single centered black hole degeneracies.

\subsection{Wall crossing and mock Jacobi forms}
For basic facts about Siegel modular forms, we refer the reader to~\cite{Freitag}.
The Igusa cusp form $\Phi_{10}$, the unique Siegel modular form of weight~$10$, is the Borcherds
(multiplicative) lift of the function~$2B(\t,z) $:
  \be \label{final2}
   \Phi_{10}(Z) \= q \zeta w \prod_{(n, \ell,m) >0}\left( 1 - q^n \zeta^\ell w^m \right)^{2C_0\left(4mn-\ell^2\right)},
   \ee
where the coefficients $C_0(\D)$ are defined in (\ref{phiC}).
Here the notation $(n,\ell,m)>0$ means that $n,\,\ell,\,m \in\Z$ with
either $m>0$ or $m=0$ and $n >0$, or $m=n=0$ and $l<0$.
In terms of the Hecke-like operators $V_{m}$, (\ref{final2})  can be rewritten in the
form
$$
\Phi_{10}(Z)= w \, \D(\t) A(\tau,z) \, \exp\left(-2\sum_{m\geq 1} B|V_{m}(\tau,z) \, w^{m}\right) \, ,
$$
where $\D(\t)$ is the weight $12$ modular form:
\[
\Delta(\tau):=q\prod_{n\geq 1}\left(1-q^n\right)^{24} = q - 24q^{2} + 252q^{3} - 1472q^{4} + 4830q^{5}+
O\left(q^{6}\right) \, .
\]
The function~$\Phi_{10}$ can also be written as the Saito-Kurokawa (additive) lift of the
Jacobi form~$ \v_{10, 1} = \Delta A$.

We are interested in the Fourier coefficients of the microscopic partition function $\Phi_{10}^{-1}$,
with respect to the three chemical potentials $(\t, z, \t')$ which are conjugate to the three $T$-duality
invariant integers~$(n,\ell,m)$.
The Igusa cusp form has double zeros at~$z=0$ and its $Sp_2(\IZ)$-images. The partition
function is therefore a meromorphic Siegel modular form   of weight
$-10$ with double poles at the divisors. As mentioned above, this meromorphicity
is responsible for the wall-crossing behavior of these functions.

The first step to analyze the Fourier coefficients \cite{Dabholkar:2011mm}
is to expand the microscopic partition function
in $w$:
  \be\label{reciproigusa}
  \frac 1{\Phi_{10}(Z)} \= \sum_{m\geq -1} \psi_m (\t,z) \, w^m  \, .
   \ee
Using \eqref{final2}, one can compute the coefficients $\psi_{m}$.
The first few   are given by \cite{Dabholkar:2011mm}
\bea\label{firstpsim}
\D\,\psi_{-1} & \= & A^{-1}\;, \nonumber \\
\D\,\psi_{0\;}  & \= & 2\,A^{-1}B\;, \nonumber \\
\D \,\psi_{1\;} &\=& \big(9\,A^{-1}B^2 \+ 3 E_4 A \big)/4\;, \\
\D \, \psi_{2\;}  & \= & \big(50 \, A^{-1} B^3 \+  48 E_4 AB \+ 10 E_6 A^2 \big)/27\;,  \nonumber  \\
\D \, \psi_{3\;}  & \= & \big(475 \,  A^{-1} B^4 \+ 886 E_4 AB^2 \+ 360 E_6 A^2 B
\+ 199 E_4^2 A^3 \big)/384\;, \nonumber \\
\D \, \psi_{4\;}  & \= & \big(51 \,  A^{-1} B^5 \+ 155 E_4 A B^3 \+  93 E_6 A^2 B^2
\+ 102 E_4^2 A^3 B \+ 31 E_4 E_6 A^4 \big)/72\; , \nonumber
\eea
where, for even $k\ge2$, the Eisenstein series $E_k$ of weight $k$ is defined as
\begin{equation} \label{Ek}
E_k(\tau):=1-\frac{2k}{B_k}\sum_{n\geq 1} \sigma_{k-1}(n)q^n
\end{equation}
with $B_k$ the $k$th Bernoulli number and $\sigma_{k-1}(n):=\sum_{d|n}d^{k-1}$.
Note that for $k \geq 4$ even the function $E_k$ is a modular form, whereas $E_2$ is a so-called quasimodular form.
The first few Eisenstein series are:
\begin{eqnarray}\label{eisen}
\nonumber   E_2 (\t) &\=& 1 \,-\, 24\,\sum_{n \geq 1} \frac{nq^n}{1- q^n} \= 1- 24q - 72q^2 - O\left(q^3\right) \, , \\
   E_4 (\t) &\=& 1 \+ 240\,\sum_{n \geq 1} \frac{n^3q^n}{1- q^n} \= 1+ 240 q + 2160q^2 + O\left(q^3\right) \, , \\
\nonumber     E_6 (\t) &\=& 1 \,-\, 504\,\sum_{n \geq 1} \frac{n^5q^n}{1- q^n} \= 1 -504q -16632q^2 - O\left(q^3\right) \,.
\end{eqnarray}
The double zero of $\Phi_{10}$ at $z=0$ is reflected by the double zeros of the denominator $A$ in
the $A^{-1} B^{m+1}$ terms in the formulas~\eqref{firstpsim}.
These meromorphic Jacobi forms were analyzed in \cite{Dabholkar:2011mm}, following
a theorem of Zwegers \cite{Zagier:2007,Zwegers:2002} who showed that the Fourier coefficients
of meromorphic Jacobi forms are related to mock modular forms.  The analysis, which we sketch below,
uniquely associates a mock Jacobi form (first systematically studied by the first author and Richter~\cite{BR}) to a meromorphic Jacobi form of the type $\psi_{m}$ above.

The first step is to define the polar part of $\psi_{m}$
\be\label{Tm}
\nonumber
 \psi^{P}_m (\tau, z) := \; \frac{p_{24}(m+1) }{\eta^{24}(\tau)} \, \sum_{s\in\Z} \, \frac{q^{ms^2 +s}\zeta^{2ms+1}}{(1 -\zeta q^s )^2} \ ,
\ee
where $p_{24}(n)$ counts the number of   partitions of an integer $n$ allowing  $24$ colors.
The function~$\psi^{P}_{m}$ is the average over the lattice $ \mathbb{Z} \tau + \Z$ of the leading behavior of the function
near the pole $z=0$
  \be \label{simplewall} \nonumber
  \frac{p_{24}(m+1)}{\eta^{24}(\tau)} \frac{\zeta}{(1-\zeta)^2}\,. \ee
The function $\psi_{m}^{P}$ is an example of an Appell-Lerch sum, and it encodes the physics of all the wall-crossings due to the decay of two-centered black holes.

The single-centered black hole degeneracies are found by subtracting the polar part from the original
meromorphic Jacobi form $\psi_{m}$. The two functions $\psi_{m}$ and $\psi_{m}^{P}$ have, by
construction, the same poles and residues, so the difference is holomorphic in $z$, and has an
unambiguous Fourier expansion. The finite or Fourier part of $\psi_{m}$
\be \nonumber
\psi_{m}^{F} := \psi_{m} - \psi_{m}^{P} \, ,
\ee
is a mock Jacobi form of index $m$. It was shown in~\cite{Dabholkar:2011mm} that the
indexed degeneracies of the single centered black hole of magnetic charge invariant~$N^{2}/2=m$,
as defined by the attractor mechanism, are Fourier coefficients of the function~$\psi_{m}^{F}$.
More precisely, we have that the microscopic indexed degeneracies~$d_{\rm micro}(n,r,m)$
corresponding to the single-centered black holes are related to the Fourier coefficients
of this function $\psi_{m}^{F} = \sum_{n,r} c(n,r) q^{n} \zeta^{r}$, as $d_{\rm micro}(n,r,m) = (-1)^{r+1} c(n,r)$,
the overall sign coming from an analysis of the fermion zero modes described in the introduction.
We  now  analyze the positivity of the numbers~$d_{\rm micro}(n,r,m)$.

We  work out the first two cases $m=1,2$.
The analysis of  \cite{Dabholkar:2011mm}  explicitly identified the mock Jacobi forms arising
as the finite parts of the meromorphic Jacobi forms for many cases.
We have the following explicit formulas
for the finite parts of the mock Jacobi forms $B^{m+1}/A$:
\begin{align}
\nonumber \left(\frac{B^2}{A}\right)^F&=E_4A-288 \mathcal{H} \, ,\\
\nonumber \left(\frac{B^3}{A}\right)^F&=3E_4AB-2E_6A^2-12^3\mathcal{H}|V_2 \, ,
\end{align}
in terms of the Hecke-like operator defined in~\eqref{defVl}, and the
function
$$
\CH(\tau,z):=\sum_{n,r} H(4n-r^{2}) q^{n} \zeta^{r},
$$
where for $N \geq 0$,  $H(N)$ denotes  the Hurwitz-Kronecker
class numbers. The function $\CH$ can be expanded as:
\begin{equation}
\nonumber
\mathcal{H}(\tau, z):=\mathcal{H}_0(\tau) \, \vartheta_{1, 0}(\tau, z)+\mathcal{H}_1(\tau) \, \vartheta_{1, 1}(\tau, z) \, ,
\end{equation}
where
\begin{align}
\nonumber
\mathcal{H}_j(\tau)&:=\sum_{n\geq 0} H(4n+3j)\, q^{n+\frac{3j}{4}} \, .
\end{align}
From work of Hirzebruch and Zagier \cite{HZ,Za} one can conclude that these functions are mock modular forms.
For later purposes, we note that  $H(0)=-1/12$ and $H(n)>0$ for $n \in \N$.

Using the formulas in \eqref{firstpsim}, we get:
\bea
 \psi_1^F &= & \frac1\Delta(3 E_4A- 648\mathcal{H}) \, , \label{m1}\\
 \psi_2^F&= & \frac{1}{3\Delta}\big(22E_4 AB-10 E_6A^2-9600 \mathcal{H}|V_2 \big) \, .\label{m2}
\eea
In the next few sections, we  prove that the coefficients
$c(n,r)$ of these two functions $\psi_1^F$, $\psi_2^F$ obey the positivity
property\footnote{For general~$m$, the $(n,r)$ Fourier coefficient of $\psi_{m}$ has an obvious
black hole interpretation for~$n \ge m$. For~$m=2$, the $n=1$ coefficient also has a black hole
interpretation, as can be seen from the table in the appendix. (The~$(n,r)=(1,1)$ coefficient of~$\psi_{2}$
are equal to the~$(2,1)$ coefficient of~$\psi_{1}$, and the~$(1,2)$ coefficient of~$\psi_{2}$
are equal to the~$(1,0)$ coefficient of~$\psi_{1}$). The general pattern remains to be
fleshed out fully.}:
\begin{equation}\label{positive}
(-1)^{r+1} c(n,r) > 0 \quad \text{for $4mn-r^{2}>0$}.
\end{equation}
The relation of the microscopic degeneracies of these five dimensional black holes~\cite{Castro:2008ys,
Banerjee:2008ag, Dabholkar:2010rm}
to their four dimensional counterparts in the $\CN=4$ theories is exactly as in the $\CN=8$ theories,
as described in \S\ref{5dlift}. The positivity in the ensemble with varying $J^{2}$ simply follows
 from the positivity of the four dimensional case. In the ensemble with fixed $J^{2}$,
one needs to show that the Fourier coefficient $c(n,r)$ of the functions $\psi^{F}_{1}$, $\psi^{F}_{2}$
obey the property $c(n,r) > c(n,r+2)$. This property will also be seen to be true in the course of
presenting the proofs below.

\section{The positivity property for $m=1$ \label{mone}}
In this section, we show (\ref{positive}) for $m=1$.  By (\ref{m1}), we have that
\[
\frac13\Delta(\tau)\psi_1^F(\tau, z)=E_4(\tau)A(\tau, z)-216\mathcal{H}(\tau, z).
\]
A direct calculation shows that $A$ has the following theta decomposition
\begin{equation} \label{Atheta}
A(\tau, z)=\frac{1}{\eta^6(\tau)} \left(\theta_0(\tau)\vartheta_{1, 1}(\tau, z)-\theta_1(\tau)\vartheta_{1, 0}(\tau, z)\right).
\end{equation}
Here we define for  $j \in \{0,1\}$
\begin{equation*}
\theta_j (\tau):=\vartheta_{1, j}(\tau, 0).
\end{equation*}
This immediately implies that
\[
\psi_1^F(\tau, z)=k_1(\tau) \vartheta_{1, 1}(\tau, z)-k_0(\tau)\vartheta_{1, 0}(\tau, z)
\]
with
\begin{align*}
k_1(\tau)&:=\frac{3}{\Delta(\tau)}\left(\frac{E_4(\tau)\theta_0(\tau)}{\eta^6(\tau)}-216\mathcal{H}_1(\tau)\right),\\
k_0(\tau)&:=\frac{3}{\Delta(\tau)}\left(\frac{E_4(\tau)\theta_1(\tau)}{\eta^6(\tau)}+216\mathcal{H}_0(\tau)\right).
\end{align*}
To prove (\ref{positive}),   we have to show that the positive Fourier coefficients of $k_1$ and $k_0$ are positive. To treat the coefficients of $k_0$, we require the following general lemma. For this define as usual
\[
(q; q)_\infty=(q)_\infty :=\prod_{\ell\geq 1}\left(1-q^\ell\right).
\]

\begin{lemma}\label{multD}
Assume that $f(q)=\sum_{n\geq 0}a(n) q^n$ satisfies $a(n)>0$ for $n\geq n_0$ $(n_0\in\N_0)$
and that for for $0\leq j\leq n_0-1$ and $\ell\in\N$ we have that $a(j+\ell n_0)>k|a(j)|$ for some $k\in\N$. Then
the function $\frac{f(q)}{(q)_\infty^k}$ has positive  coefficients for $ n \geq n_0$.
\end{lemma}

\begin{proof}
By normalizing and  splitting the coefficients of $f$ into residue classes $\mypmod{n_0}$, we may assume that
\[
f(q)=-1+\sum_{n\geq 1} a(n)q^n
\]
satisfies $a(n)>k$. We may view this function as  the $r=1$ case of the more general family of functions
\begin{equation}\label{fr}
f_r(q)=-1+\sum_{n\geq 1} a_r(n)q^n
\end{equation}
that satisfies $a_r(n)>0$ for $1\leq n\leq r$ and $a_r(n)>k$ for $n\geq r$. To be more precise, we  define the functions $f_r$ inductively as
\[
f_{r+1}(q):=\frac{1}{(1-q^r)^k}f_r(q).
\]
Note that  the coefficients of $f_{r+1}$ that are not divisible by $r$ may be bounded below by those of $f_r$. The remaining coefficients have the shape of $f_1$ (with $q\mapsto q^r$) thus the claim follows inductively as soon as we show it for $r=1$. for this recall that
\[
\frac{1}{(1-q)^k}=\sum_{\ell\geq 0}\binom{\ell+k-1}{k-1}q^\ell.
\]
Thus in $f_2$, the first coefficient equals $a(1)-k>0$ and for $n>1$, the $n$-th coefficient equals
\[
\sum_{0\leq j\leq n-1}\binom{j+k-1}{k-1} a(n-j) -\binom{n+k-1}{k-1}\geq a(n) +\binom{n+k-2}{k-1}k-\binom{n+k-1}{k-1}>k.
\]
This yields the claim of the lemma.
\end{proof}

To apply Lemma \ref{multD}, we write
\[
\frac{q}{48}k_0(\tau)\Delta(\tau)=-1+\sum_{n\geq 1}a(n) q^n.
\]
Since the coefficients of $\theta_1/\eta^6$ are non-negative, and the class numbers~$H(n)>0$ for
$n>0$, we may, using (\ref{Ek}), bound the coefficients $a(n)$ for $n > 1$  by
\[
a(n)\geq    15\sigma_3(n)>24,
\]
and we can check that~$a(1)>24$.
Thus we directly obtain from Lemma \ref{multD} with $k=24$ and $n_0=1$ that the~$n>0$ coefficients of~$k_0$ are positive.

We next turn to $k_1$. It is clearly enough to show that for  $n>0$  the $n$th coefficient of $\frac{q^{\frac14}}{24}\Delta k_1$ is positive. This may be bounded from below by
\begin{equation}\label{s against H}
10\sigma_3(n)-9H(4n-1) .
\end{equation}
Clearly
\[
\sigma_3(n)\geq n^3.
\]
Moreover, it is not hard to show that
\begin{equation} \label{classtrivial}
H(n)<n.
\end{equation}
Thus \eqref{s against H} may be bounded from below by
\[
10 n^3-9(4n-1),
\]
which is positive for $n\geq 2$. The claim then follows since
$$
\frac{1}{24}\Delta (\tau) k_1 (\tau)=
q^{-\frac14} \left(1+176q+O\left( q^2\right) \right).
$$

\section{The positivity property for $m=2$ \label{mtwo}}
In this section we prove (\ref{positive}) for $m=2$ relying on the
Circle Method and asymptotic formulas as shown by Manschot and the
first author \cite{BM}.
In the next section we will present a second, more elementary proof. Both proofs
make use of the theta decompositions of the functions involved.
\subsection{Certain theta decompositions}
We first show the following theta decomposition
\begin{equation}\label{ThetaD}
\frac{1}{\eta^6(\tau)}\left(11E_4(\tau) A(\tau, z) B(\tau, z)-5E_6(\tau) A^2(\tau, z)\right)
=\sum_{0\le j\le 3} h_j (\tau)\vartheta_{2, j}(\tau, z)
\end{equation}
with
\begin{align*}
h_0(\tau)&:=-\frac{1}{\eta^{18}(\tau)}\left(\theta_0(2\tau)\theta_1(\tau)f_0(\tau)+\theta_1(2\tau)\theta_0(\tau)f_1(\tau)\right),\\
h_1(\tau)&:=h_3(\tau)=\frac{1}{2\eta^{18}(\tau)}\theta_1\left(\frac{\tau}{2}\right)\left(f_0(\tau)\theta_0(\tau)+f_1(\tau)\theta_1(\tau)\right),\\
h_2(\tau)&:=-\frac{1}{\eta^{18}(\tau)}\left(\theta_1(2\tau)\theta_1(\tau)f_0(\tau)+\theta_0(2\tau)\theta_0(\tau)f_1(\tau)\right).
\end{align*}
Here
\begin{align} \label{definef0}
f_0(\tau)&:=264 \theta_1'(\tau)E_4(\tau)+\left(5E_6(\tau)-11E_2(\tau)E_4(\tau)\right)\theta_1(\tau),\\
\label{definef1}
f_1(\tau)&:=264 \theta_0'(\tau)E_4(\tau)+\left(5E_6(\tau)-11E_2(\tau)E_4(\tau)\right)\theta_0(\tau),
\end{align}
where the prime denotes $\frac{1}{2 \pi i} \frac{d}{d\t}$.
In particular the above representations imply that
\begin{align}
h_0(\tau) = & -q^{-\frac14} \Big(228+39096 q+1205988 q^2+21844152 q^3+278145540 q^4+2742795528 q^5 \nonumber \\
&\qquad \qquad +22290285288 q^6 +155617854912 q^7+960737806812 q^8 + O\big(q^{9}\big)\Big)\label{h0},\\
h_1(\tau)=&q^{-\frac38}\Big(108+15420 q+669192 q^2+14367108 q^3+198499812 q^4+2050094076 q^5
 \nonumber \\
&\qquad\qquad +17163958500 q^6+122388860268 q^7+767849126316 q^8+O\big(q^{9}\big)\Big)\label{h1},\\
h_2(\tau)=&-q^{-\frac34}\Big(-6-4020 q+81390 q^2+4075236q^3+72603588q^4+ 856025184q^5
 \nonumber \\
& \qquad\qquad + 7805050218q^6 + 59195535780q^7 + 389556957342q^8 + O\big(q^{9}\big)\Big)\label{h2}.
\end{align}
To prove \eqref{ThetaD}, we first recall the theta decomposition (\ref{Atheta}) of $A$.
To find the theta decomposition of $B$, we write
\begin{equation}\label{BTheta}
B(\tau, z)=g_0(\tau)\vartheta_{1, 0}(\tau, z) + g_1(\tau) \vartheta_{1, 1}(\tau, z).
\end{equation}
Since $B$ is a Jacobi form of weight $0$ and index $1$, the functions $h_j$ are components of
a 2-dimensional vector valued modular form which one can show lies in a 1-dimensional space. From this
one may conclude that
\begin{align}
g_0(\tau)&=\frac{1}{\eta^6(\tau)}\left(24\theta_1'(\tau)-E_2(\tau)\theta_1(\tau)\right)\label{g0},\\
g_1(\tau)&=\frac{1}{\eta^6(\tau)}\left(-24\theta_0'(\tau)+E_2(\tau)\theta_0(\tau)\right)\label{g1}.
\end{align}
This yields that
\begin{equation}\label{Thetasecond}
11 E_4(\tau)B(\tau, z)-5E_6(\tau)A(\tau, z)=\frac{1}{\eta^6(\tau)}\left(f_0(\tau)\vartheta_{1, 0}(\tau, z)-f_1(\tau)\vartheta_{1, 1}(\tau, z)\right),
\end{equation}
with $f_0$ and $f_1$ defined in (\ref{definef0}) and (\ref{definef1}), respectively.
Multiplying \eqref{Atheta} and \eqref{BTheta} and using that
\begin{align*}
\vartheta_{1, 0}^2(\tau, z)&=\theta_0(2\tau)\vartheta_{2, 0}(\tau, z)+\theta_1(2\tau)\vartheta_{2, 2}(\tau, z),\\
\vartheta_{1, 1}^2(\tau, z)&=\theta_1(2\tau)\vartheta_{2, 0}(\tau, z)+\theta_0(2\tau)\vartheta_{2, 2}(\tau, z),\\
\vartheta_{1, 0}(\tau, z)\vartheta_{1, 1}(\tau, z)&=\frac12\theta_1\left(\frac{\tau}{2}\right)\left(\vartheta_{2, 1}(\tau, z)+\vartheta_{2, 3}(\tau, z)\right)
\end{align*}
then easily gives the claimed representations for the functions $h_j$.

We next turn to the contribution coming from the  class numbers. Using the definition of $V_2$, we obtain that
\be \label{HV21}
\mathcal{H}(\tau, z) |V_2=\sum_{0\le j\le 3} \mathcal{F}_j(\tau)\vartheta_{2, j}(\tau, z)
\ee
with
\[
\mathcal{F}_j(\tau):=\sum_{\Delta \ge 0} c\left(\frac{\Delta+j^2}{8}, j\right)q^{\frac{\Delta}{8}}.
\]
Here $c(n, r)=0$ unless $n\in\N_0$ in which case it is defined by
\[
c(n, r):=\sum_{d|(n, r, 2) \atop d >0} dH\left(\frac{8n-r^2}{d^2}\right).
\]
The first few Fourier coefficients of the functions~$\mathcal{F}_{j}$ are given by
\bea \nonumber
\mathcal{F}_{0}(\t)  & \= & -\frac14 + q + \frac52 q^2 + 2q^3 + 5q^4 + 2q^5 + 6q^6 + 4q^7 + \frac{13}{2}q^8 + 3q^9 + O\left(q^{10}\right) \, , \\
\nonumber
\mathcal{F}_{1}(\t) & \= & q^{-\frac18} \left( q + 2q^2 + 3q^3 + 3q^4 + 4q^5 + 5q^6 + 4q^7 + 5q^8 + 7q^9 + O\left(q^{10}\right) \right) \, , \\
\nonumber
\mathcal{F}_{2}(\t) & \= & q^{-\frac12} \left( \frac12 q + 2q^2 + 2q^3 + 4q^4 + \frac52 q^5 + 6q^6 + 2q^7 + 8q^8 + 4q^9 + O\left(q^{10}\right) \right) \, .
\eea

Using the notation above we now aim to show that  for  $n>0$ the $n$th  coefficient of
\begin{equation} \label{showpos}
(-1)^{j+1} \frac{1}{\eta^{18}} \left( h_j-4800\frac{\mathcal{F}_j}{\eta^6}  \right)
\end{equation}
is positive.

\subsection{Asymptotic formulas for the coefficients of  $h_j$}
We write
\[
h_j^\ast(\tau):=q^{\alpha_j}h_j(\tau)=\sum_{n\geq 0} \alpha_j(n)q^n
\]
with $\alpha_0:=\frac14,\ \alpha_1=\alpha_3:=\frac38,\text{and } \alpha_2:=\frac34$.
The goal of this section is to asymptotically bound the coefficients $\alpha_j (n)$.

\begin{proposition}\label{modularpart}
We have that
\[
\alpha_j(n)=m_j(n)+e_{j1}(n)+e_{j2}(n)
\]
with
\begin{align*}
m_j(n)&:=(-1)^{j+1}2^{-\frac32}3^{\frac94}\pi\left(n-\alpha_j\right)^{-\frac54} I_{\frac52}\Big(2\pi \sqrt{3\left(n-\alpha_j\right)}\Big),\\
|e_{j1}(n)|&<216 \pi\left(n-\alpha_j\right)^{-\frac34} I_{\frac52}\Big(\pi \sqrt{6\left(n-\alpha_j\right)}\Big),\\
|e_{j2}(n)|&<47352\pi\left(n-\alpha_j\right)^{-\frac34}.
\end{align*}
Here $I_\ell$ is the usual I-Bessel function of order $\ell$.
\end{proposition}

\begin{proof}
We use the usual set up for the Circle Method. To be more precise, we assume that $0\leq h<k$ with $(h, k)=1,\ hh'\equiv -1\mypmod{k}$ and $z\in\C$ with $\text{Re}(z)>0$.
Using this notation, we have the transformation law
\[
h_j\left(\frac1k(h+iz)\right)=z^{\frac32}\sum_{0\le \ell \le 3}\chi_{j, \ell}(h. k)h_\ell\left(\frac1k\left(h'+\frac{i}{z}\right)\right),
\]
where $\chi_{j, \ell}$ is a multiplier satisfying
\[
\left|\chi_{j, \ell}(h, k)\right|\leq 1,\quad
\chi_{j, \ell}(0, 1) =-\frac{i^{j\ell}}{2}.
\]
Moreover from \eqref{h0}, \eqref{h1}, and \eqref{h2} we obtain that
\[
h_j(\tau)= q^{-\alpha_j}\left(\delta_j+O\left(q\right)\right)
\]
with  $\alpha_0=1/4, \alpha_1= \alpha_3=3/8, \alpha_2=3/4, \delta_0=-228, \delta_1=\delta_3=108, \text{and } \delta_2=6$. Using the classical Circle Method (see e.g. \cite{RZ}) then gives that
\[
\alpha_j(n)=2\pi \sum_{0\le \ell \le 3} \delta_\ell \sum_{h, k}
\frac1k e^{\frac{2\pi i}{k}\big(h\left(\alpha_j-n\right)-h'\alpha_\ell\big)}
\chi_{j, \ell}(h, k)\left(\frac{n-\alpha_j}{\alpha_\ell}\right)^{-\frac54} I_{\frac52}\left(\frac{4\pi\sqrt{\alpha_\ell\left(n-\alpha_j\right)}}{k}\right),
\]
where the sum runs over all $0\leq h<k$ with $(h, k)=1$. Using that for $r\in\R$
\[
I_r(x)\sim\frac{e^x}{\sqrt{2\pi x}}\qquad (x\to\infty)
\]
gives that the dominant term arises from $k=1$ and $\ell=2$ and is given by $m_j(n)$ as stated in the theorem.

The remaining sums may be bounded by
\[
2\pi \left(n-\alpha_j\right)^{-\frac54}\sum_{0\le \ell \le 3}|\delta_\ell|\alpha_\ell^{\frac54} \sum_k^\ast  I_{\frac52}\left(\frac{4\pi\sqrt{\alpha_\ell\left(n-\alpha_j\right)}}{k}\right),
\]
where $\sum_k^\ast$ denotes the sum on $k$ with the $k=1$ term dropped in the case $\ell=2$.
 We first split off those $k$ for which $k\leq \sqrt{n-\alpha_j}$. It is easy to see that
$\frac{\sqrt{\alpha_\ell}}{k}$ is maximized for $k=1$ and $\ell=1, 3$ in which case it equals $\frac{\sqrt{3}}{2\sqrt{2}}$. Using that $I_{\frac52}(x)$ is increasing, the contribution from $k\leq \sqrt{n-\alpha_j}$
may be estimated against
\[
2\pi \left(n-\alpha_j\right)^{-\frac34}I_{\frac52}\left(\pi\sqrt{6(n-\alpha_j)}\right)\sum_{0\le \ell \le 3}|\delta_\ell|\alpha_\ell^{\frac54}  .
\]
Bounding the sum on $\ell$ gives the bound for $e_{j1}(n)$ as stated in the theorem.

Using the integral comparison criterion, the terms from $k>\sqrt{n-\alpha_j}$ can be bounded by
\[
2\pi\left(n-\alpha_j\right)^{-\frac54}\sum_{0\le \ell \le 3}|\delta_\ell|\alpha_\ell^{\frac54} \int_{\sqrt{n-\alpha_j}}^\infty  I_{\frac52}\left(\frac{4\pi\sqrt{\alpha_\ell(n-\alpha_j)}}{x}\right)dx.
\]
Using the series representation of the Bessel function it is not hard to see that $\frac{I_\ell(x)}{x^\ell}$ is monotonically increasing. Therefore we may estimate the integral against
\[
I_{\frac52}\left(4\pi\sqrt{\alpha_\ell}\right)\left(n-\alpha_j\right)^{\frac54}\int_{\sqrt{n-\alpha_j}}^\infty x^{-\frac52}dx.	
\]
Explicitly evaluating the integral and estimating  the sum on $\ell$ gives the bound for $e_{j2}(n)$ as stated in the theorem.
\end{proof}

\subsection{Bounding the class number contribution}
In this section we bound for $j=0,1$ the contribution
\[
\frac{4800}{\eta^6(\tau)}\mathcal{F}_j(\tau)
=: \sum_{n \geq 0} \beta_j(n)q^{n- \alpha_j}.
\]
Note that in the case $j=2$, the coefficients of the class number function will later be ignored and are thus not considered in this section.

In the case $j=1$, we   relate the coefficients $\beta_1(n)$ to the coefficients of  a function studied by the first author and Manschot \cite{BM}.
To be more precise, we define
$$
\frac{4800}{\eta^6(\tau)} \mathcal{H}_1(q) =: \sum_{n \geq 0} \gamma(n)q^{n-\frac{1}{2}}.
$$
Moreover we denote by $p_6(n)$ the number of partitions of $n$ allowing $6$ colors.
Note that
$$
\frac{1}{(q)_{\infty}^6}= \sum_{n \geq 0} p_6(n)q^n.
$$
Using that $H(n)>0$ for $n \in \N$ and that $p_6(n)$ is monotonically increasing,
 it is not hard to show
\begin{lemma} \label{relatefunctions}
We have
\begin{eqnarray*}
-\beta_0(n)&\leq & 1200 p_{6}(n),\\
\beta_1(n)&\leq& \gamma(2n).\\
\end{eqnarray*}
\end{lemma}
We first bound the coefficients $p_6(n)$.
\begin{lemma}
We have
$$
p_6(n)= e_1(n)+e_2(n)
$$
with
\begin{eqnarray*}
e_1(n)&<&  \frac{\pi}{8} \left( n-\frac14\right)^{-\frac32} I_4 \left(2\pi \sqrt{ n-\frac14    }\right),\\
e_2(n)&<& \pi\left( n-\frac14\right)^{-\frac32}.
\end{eqnarray*}
\end{lemma}
\begin{proof}
Firstly we may show by the classical Circle Method that
$$
p_6(n)= \frac{\pi}{8} \left( n-\frac14\right)^{-2}
\sum_{h, k}
\frac1k
\chi(h, k) I_{4}\left(\frac{2\pi\sqrt{n-\frac14}}{k}\right),
$$
where the sum runs over all $0\leq h<k$ with $(h, k)=1$.  Now the claim follows as in the proof of Lemma \ref{modularpart}.
\end{proof}

\begin{lemma}\label{lemma class}
We have the bounds
\[
\gamma(n) \leq
\rho_{1}(n)+\rho_{2}(n)+\rho_{3}(n)+\rho_{4}(n)+\rho_{5}(n)+\rho_{6}(n)
\]
with
\begin{align*}
\rho_{1}(n) &< 400 \pi\left(4n-2\right)^{-\frac34}
I_{\frac52}\left(\pi\sqrt{4n-2}\right),\\
\rho_{2}(n) &<13603\pi\left(4n-2\right)^{-\frac34},\\
\rho_{3} (n)&<541\pi \left(4n-2\right)^{-1} I_3\left(\pi\sqrt{4n-2}\right),\\
\rho_{4} (n)&<10330 \pi\left(4n-2\right)^{-\frac34},\\
\rho_{5} (n)&<244 \pi\left(4n-2\right)^{-\frac54} I_{\frac72}\left(\pi\sqrt{4n-2}\right),\\
\rho_{6}(n) &<2519 \pi\left(4n-2\right)^{-\frac34}.
\end{align*}
\end{lemma}

\begin{proof}
In \cite{BM} the first author and Manschot proved an exact formula for $\gamma_j(n)$.
We employ this formula and bound all occurring Kloosterman sums trivially to obtain
\[
 \gamma(n)=\mu_{1}(n)+\mu_{2}(n)+\mu_{3}(n)
\]
with
\begin{align*}
|\mu_{1}(n)|&<800 \pi(4n-2)^{-\frac54}\sum_{k=1}^\infty I_{\frac52}\left(\frac{\pi}{k}\sqrt{4n-2}\right),\\
|\mu_{2}(n)|&<\frac{4800}{\sqrt{2}}(4n-2)^{-\frac32}\sum_{k=1}^\infty \sqrt{k}I_3\left(\frac{\pi}{k}\sqrt{4n-2}\right),\\
|\mu_{3}(n)|&<\frac{600}{\pi}(4n-2)^{-\frac74}\sum_{k=1}^\infty \frac{1}{k}\sum_{{\ell\in\{0, 1\}\atop{-k<g\leq k}}\atop{g\equiv \ell\mypmod{2}}}\left|\mathcal{I}_{k, g}(n)\right|.
\end{align*}
Here
\[
\mathcal{I}_{k, g}(n):=\int_{-1}^1 f_{k, g}\left(\frac{u}{2}\right)I_{\frac72}\left(\frac{\pi}{k}\sqrt{(4n-2)\left(1-u^2\right)}\right)\left(1-u^2\right)^{\frac74}du
\]
with
\[
f_{k, g}(u) :=
\begin{cases}
\frac{\pi^2}{\sinh^2\left(\frac{\pi u}{k}-\frac{\pi ig}{2k}\right)}&\quad\text{ if } g\not\equiv 0\mypmod{2k},\\
\frac{\pi^2}{\sinh^2\left(\frac{\pi u}{k}\right)}-\frac{k^2}{u^2}&\quad\text{ if } g\equiv 0\mypmod{2k}.
\end{cases}
\]
The terms in $\mu_{1}(n)$ may now be bounded as before, splitting the sum on $k$ at $\frac12\sqrt{4n-2}$, giving the estimates for $\rho_{1}(n)$ and $\rho_{2}(n)$ as stated in the lemma.
 For the terms in $\mu_{2}(n)$ we proceed similarly to obtain the bounds for $\rho_{3}(n)$ and $\rho_{4}(n)$ as stated in the lemma.

  To finish the proof, we have to estimate the terms in $\mu_{3}(n)$. First by the proof of Lemma 3.2 in \cite{BM},
   we may for $-k<g\leq k$ bound $f_{k, g}(u)\leq h_{k, g}$ with
\[
h_{k, g} :=
\begin{cases}
\frac{k^2}{g^2}&\quad\text{ if } -k<g\leq k, g\neq 0,\\
1              &\quad\text{ if } g=0.
\end{cases}
\]
Thus
\[
\mathcal{I}_{k, g}(n)\leq 2 h_{k, g} I_{\frac72}\left(\frac{\pi}{k}\sqrt{(4n-2)}\right).
\]
This gives the estimate
\[
|\mu_{3}(n)|\leq \frac{1200}{\pi}(4n-2)^{-\frac74}\sum_{k=1}^\infty \frac{I_{\frac72}\left(\frac{\pi}{k}\sqrt{4n-2}\right)}{k}
\sum_{{\ell\in\{0, 1\}\atop{-k<g\leq k}}\atop{g\equiv \ell\mypmod{2}}}h_{k, g}.
\]
It is not hard to see that   the sum on $\ell$ and $g$ may be bounded by $4k^{2}$
and proceed as before yielding the estimates for $\rho_{5}(n)$ and $\rho_{6}(n)$ as given in the lemma.
\end{proof}

\subsection{The final estimates}
In this section we finish the proof of (\ref{showpos}) by comparing the asymptotic
growth of the functions involved.

We  throughout use the easily verified rough bound for $x\geq 20$
\begin{equation}\label{rough bound}
\frac{4 e^x}{5\sqrt{2\pi x}}\leq I_{\frac52}(x)\leq \frac{e^x}{\sqrt{2\pi x}}.
\end{equation}
Note that the upper bound holds true for all $ x \geq 0$.

We start with the simplest case $j=2$. We first bound the  coefficients of $h_2$ and begin by  comparing  the contributions coming from the error term $e_{22}(n)$
with the main term $m_2(n)$:
\[
\left\lvert\frac{e_{22}(n)}{m_2(n)}\right\rvert < 11308 \frac{\left(n-\frac34\right)^{\frac12}}{I_{\frac52}\left(2\pi\sqrt{3\left(n-\frac34\right)}\right)}.
\]
Using that $\frac{I_{\frac52}(x)}{x}$ is monotonically increasing, we obtain that for $n\geq 4$
\begin{equation}\label{E2M}
\left\lvert\frac{e_{22}(n)}{m_2(n)}\right\rvert
 < 0.001.
\end{equation}
We next turn to the contribution coming from $e_{21}(n)$. We bound
\[
\left\lvert\frac{e_{21}(n)}{m_2(n)}\right\rvert <  52\frac{\left(n-\frac34\right)^{\frac12}I_{\frac52}\left(\pi\sqrt{6\left(n-\frac34\right)}\right)}{I_{\frac52}\left(2\pi \sqrt{3\left(n-\frac34\right)}\right)}.
\]
Using \eqref{rough bound} then yields that
\[
\left\lvert\frac{e_{21}(n)}{m_2(n)}\right\rvert <   78\left(n-\frac34\right)^{\frac12}e^{-\pi \left(\sqrt{2}-1\right)\sqrt{6\left(n-\frac34\right)}}.
\]
Since $\frac{e^{ax}}{x}$ is monotonically increasing for $x>\frac1a$, we obtain that for $n\geq 4$
\begin{equation}\label{E1M}
\left\lvert\frac{e_{21}(n)}{m_2(n)}\right\rvert <  0.45.
\end{equation}
Combining \eqref{E2M} and \eqref{E1M} gives that for $n\geq 4$ the $n$th Fourier coefficient of $h_2$ is negative. Then employing \eqref{h2} gives that for $n\geq 2$ the $n$th Fourier coefficient of $h_2$ is negative. Thus for $n\geq 2$ the $n$th Fourier coefficient of
$h_2-\frac{4800\mathcal{F}_2}{\eta^6}$
is negative. To finish the proof, we aim to apply Lemma \ref{multD}. For this, we note that
\begin{align} \label{smallh2}
h_2(\tau) -\frac{4800\mathcal{F}_2(\tau)}{\eta^6(\tau)}
= & q^{-\frac34}\left( 6 + 1620q - 105390q^2  - 407236q^3- 73174788q^4 + O\left(q^{5}\right) \right).
\end{align}
We apply Lemma \ref{multD} with $n_0=2$ and $k=18$. Inspecting the first $3$ coefficients by hand, we are
 left   to show that  for $n \geq 4$, the $n$th Fourier coefficient of this function  is in absolute value bigger than
$18 \cdot 1620$.
From the above calculations it immediately follows that this coefficient   may be bounded by $0.5|m_2(n)|$.
Using that $|m_2(n)|$ is monotonically increasing, then easily gives that for $n \geq 4$ this satisfies the claimed bound.
Therefore we have shown that
(\ref{showpos}) holds true for $j=2$.

We next turn to the case $j=1$. In this case we have to take the class number contribution into account.
As before we may show that for $n\geq 10$
\[
\frac{|e_{11}(n)|+|e_{12}(n)|}{|m_1(n)|}<0.02.
\]
To estimate the class number contribution, we use Lemmas \ref{relatefunctions} and \ref{lemma class} and bound
\[
\frac{|\rho_{2}(2n)|+|\rho_{4}(2n)|+|\rho_{6}(2n)|}{|m_{1}(n)|}<1328\frac{\left(n-\frac38\right)^{\frac54}}{\left(n-\frac14\right)^{\frac34}I_{\frac52}\left(2\pi\sqrt{3\left(n-\frac38\right)}\right)}.
\]
Since the right hand side is monotonically decreasing as a function of $n$, we obtain that for $n\geq 10$
\[
 \frac{|\rho_{2}(2n)|+|\rho_{4}(2n)|+|\rho_{6}(2n)|}{|m_{1}(n)|}
<0.01.
\]
Next we see that
\begin{equation}\label{F1M}
\left\lvert\frac{\rho_{1}(2n)}{m_1(n)}\right\rvert <20.1\frac{\left(n-\frac38\right)^{\frac54}I_{\frac52}\left(2\pi\sqrt{2 \left(n-\frac14\right)}\right)}{\left(n-\frac14\right)^{\frac34}I_{\frac52}\left(2\pi\sqrt{3\left(n-\frac38\right)}\right)}.
\end{equation}
Similarly
\[
\left\lvert\frac{\rho_{3}(2n)}{m_1(n)}\right\rvert
<16.2
\frac{\left(n-\frac38\right)^{\frac54}I_3\left(2\pi\sqrt{2 \left(n-\frac14\right)}\right)}{\left(n-\frac14\right)I_{\frac52}\left(2\pi\sqrt{3\left(n-\frac38\right)}\right)}.
\]

Using that for $\ell\geq 0$
\begin{equation}\label{compare index}
I_{\frac52+\ell}(x)\leq \left(\frac{x}{2}\right)^\ell I_{\frac52}(x)
\end{equation}
yields that
\begin{equation}\label{F3M}
\left\lvert\frac{\rho_{3}(2n)}{m_1(n)}\right\rvert   < 34.2
\frac{\left(n-\frac38\right)^{\frac54}I_{\frac52}\left(2\pi\sqrt{2 \left(n-\frac14\right)}\right)}{\left(n-\frac14\right)^{\frac34}I_{\frac52}\left(2\pi\sqrt{3\left(n-\frac38\right)}\right)}.
\end{equation}
Finally
\[
\left\lvert\frac{\rho_{5}(2n)}{m_1(n)}\right\rvert
<4.4
\frac{\left(n-\frac38\right)^{\frac54}I_{\frac72}\left(2\pi\sqrt{2 \left(n-\frac14\right)}\right)}{\left(n-\frac14\right)^{\frac54}I_{\frac52}\left(2\pi\sqrt{3\left(n-\frac38\right)}\right)}.
\]
Using again \eqref{compare index} gives that
\begin{equation}\label{F5M}
\left\lvert\frac{\rho_{5}(2n)}{m_1(n)}\right\rvert
<19.6
\frac{\left(n-\frac38\right)^{\frac54}I_{\frac52}\left(2\pi\sqrt{2 \left(n-\frac14\right)}\right)}{\left(n-\frac14\right)^{\frac34}I_{\frac52}\left(2\pi\sqrt{3\left(n-\frac38\right)}\right)}.
\end{equation}

Combining \eqref{F1M}, \eqref{F3M}, and \eqref{F5M} and then  using \eqref{rough bound}  gives that
\[
 \frac{|\rho_{1}(2n)|+|\rho_{3}(2n)|+|\rho_{5}(2n)|}{|m_{1}(n)|} < 102.3\frac{\left(n-\frac38\right)^{\frac32}}{\left(n-\frac14\right)} e^{-2\pi \left(\sqrt{3\left(n-\frac38\right)}-\sqrt{2\left(n-\frac14\right)}\right)}.
\]
For $n\geq 10$, we may bound
\[
2\pi\left(\sqrt{3}-\sqrt{2\frac{n-\frac14}{n-\frac38}}\right)>1.93.
\]
Thus
\[
\frac{|\rho_{1}(2n)|+|\rho_{3}(2n)|+|\rho_{5}(2n)|}{|m_{1}(n)|}  < 102.3\sqrt{n-\frac38} e^{-1.93\sqrt{n-\frac38}}.
\]
From this we obtain as before that for $n\geq 10$
\[
\frac{|\rho_{1}(2n)|+|\rho_{3}(2n)|+|\rho_{5}(2n)|}{|m_{1}(n)|}  <0.8.
\]
Combining the above,   we have shown that for $n\geq 10$ the $n$th coefficient of
\[
h_1-\frac{4800\mathcal{F}_1}{\eta^6}
\]
may be bounded from  below by $0.17m_1(n)$ and are thus in particular positive. Using that
\begin{align*}
& h_1(\tau)-\frac{4800\mathcal{F}_1(\tau)}{\eta^6(\tau)} \\
& \quad =  q^{-\frac38}\Big( 108 + 10620q + 630792q^2 + 14165508q^3 + 197669412q^4 + 2047146876q^5  \\
& \quad \qquad + 17154588900q^6 + 122361457068q^7 + 767774193516q^8  + 4336015791756q^9
+ O\left(q^{10}\right) \Big) \,
\end{align*}
gives that all Fourier coefficients of this function are positive. Thus for $j=1,3$, also the coefficients of (\ref{showpos})
are positive.

We finally consider the case $j=0$.  As in the case $j=2$, we may bound for $n\geq 4$
\[
\frac{|e_{01}(n)|+|e_{02}(n)|}{|m_0(n)|}<0.32.
\]
Next, we use Lemmas \ref{relatefunctions} and \ref{lemma class} and bound the error terms separately.
Firstly
 $$
1200\left| \frac{e_{2}(n)}{m_0(n)}   \right|< 286.6  \frac{1}{\left(n-\frac14\right)^{\frac14}I_{\frac52}\left(2\pi\sqrt{3\left(n-\frac14\right)}\right)}.
 $$
 Using that the right hand side is monotonically decreasing as a function of $n$,  we obtain that
  for $n \geq 4$
 $$
 1200\left| \frac{e_{2}(n)}{m_0(n)}   \right|< 0.01.
 $$
 Finally we bound
$$
1200\left| \frac{e_{1}(n)}{m_0(n)}   \right|< 35.9\frac{I_{4}\left(2\pi\sqrt{n-\frac14} \right)
}{\left(n-\frac14\right)^{\frac14}
I_{\frac52}\left(2\pi\sqrt{3\left(n-\frac14\right)}\right)
}.
$$
 Using (\ref{compare index}) we then obtain
$$
1200\left| \frac{e_{1}(n)}{m_0(n)}   \right|< 200\frac{\left(n-\frac14\right)^{\frac12} I_{\frac52}\left(2\pi\sqrt{n-\frac14} \right)}{I_{\frac52}\left(2\pi\sqrt{3\left(n-\frac14\right)}\right)}.
$$
 Inserting (\ref{rough bound}) then gives
 $$
1200 \left| \frac{e_{1}(n)}{m_0(n)}   \right|< 330\left(n-\frac14\right)^{\frac12} e^{ - 2\pi \left( \sqrt{3}-1\right) \sqrt{n-\frac14}  }.
$$
 Using that the right hand side is monotonically decreasing, we obtain that for $n \geq 4$
 $$
1200\left| \frac{e_{1}(n)}{m_0(n)}   \right|<0.09.
 $$

Thus we have shown that for $n\geq 4$ the $n$th coefficient of
\[
-h_0+\frac{4800 \mathcal{F}_0}{\eta^6}
\]
may be bounded by $0.58|m_0(n)|$ and is in particular positive.
We now apply Lemma \ref{multD} with $n_0=0$ and $k=18$.
We compute  that
\begin{equation*}
 h_0(\tau)-\frac{4800\mathcal{F}_0(\tau)}{\eta^6(\tau)}   =  q^{-\frac14}\left( 972 - 36696q - 1214388q^2
+ O\left(q^{3}\right) \right) \, .
\end{equation*}
Thus we have to show that for $n \geq1$ the $n$th coefficient  of this function is bigger than $18\cdot 972$. A direct
inspection of the Fourier coefficients gives that this is true for $n \leq 2$. For $n  \geq 3$ the above calculations
give that the absolute value of the  $n$th coefficient of this function may be bounded by $0.58|m_0(n)|$ and the claim follows, again using
that $|m_0(n)|$ is monotonically increasing as a function of $n$.

\section{An alternative proof for $m=2$ \label{alter}}

For the second approach, we first estimate the growth of the coefficients of the Jacobi form
\be \label{myhdef}
11E_4(\tau) A(\tau,z)B(\tau,z) - 5E_6(\tau) A^2(\tau,z)
 =: \sum_{0 \le j \le 3} \myh_{j}(\t) \, \vartheta_{2,j}(\t,z) \, ,
\ee
to then show that for $n>0$ the $n$th coefficient of
\begin{equation} \label{xiD}
(-1)^{j+1} \frac{1}{\Delta} \left( \myh_j - 4800\mathcal{F}_j \right)
\end{equation}
is positive.

From equation~\eqref{ThetaD}, we get:
\bea
\label{h0def} \myh_{0} (\t)
& = & - \frac{1}{\eta^{12}(\t)} \left( \theta_{0}(2\t) \, \theta_{1}(\t) \, f_{0} (\t)
\+\theta_{1}(2\t) \, \theta_{0}(\t) \, f_{1}(\t) \right) \, , \\
\label{h1def} \myh_{1} (\t) \= \myh_{3} (\t) & = &  \frac{1}{2 \, \eta^{12}(\tau)} \theta_{1}\left(\frac{\t}{2}\right)
\left(f_0(\tau)\theta_0(\tau) + f_1(\tau) \theta_1(\tau) \right) \, , \\
\label{h2def} \myh_{2} (\t) & = & - \frac{1}{\eta^{12}(\t)} \left( \theta_{1}(2\t) \, \theta_{1}(\t) \, f_{0} (\t)
\+\theta_{0}(2\t) \, \theta_{0}(\t) \, f_{1}(\t) \right) \, ,
\eea
with $f_0$ and $f_1$ defined as in (\ref{definef0}) and (\ref{definef1}). %
We have thus managed to express the terms of the theta decomposition of~\eqref{myhdef} in terms of
Eisenstein series, theta functions, and eta functions.

We now address the issue of positivity of various Fourier coefficients.
We find many functions whose coefficients are all positive except for the first few ones,
and these are then   multiplied by theta series and negative powers of the eta function.
In order to handle such products, we  use  Lemma \ref{multD} and the following lemma which
 deals with multiplication by  theta series and whose proof is straightforward. 
\begin{lemma}\label{thetamult}
Assume that $f(q)=\sum_{n\geq 0} a(n)q^n$ satisfies $a(n)>0$ for $n>n_0$ for some $n_0\in\N_0$. Then for $\lambda\in\{0, 1\}$ and $m\in\N$ the $n$th coefficient of $q^{-m\lambda^2/4}f(q)\theta_\lambda(m\tau)$
is at least $\delta_\lambda a(n)-2\sum_{0\leq j\leq n_0}|a(j)|$, where $\delta_0=1$ and $\delta_1=2$.
Moreover if for some $n_1\in\N$ we have that
$\delta_\lambda a(n)>2\sum_{0\leq j\leq n_0}|a(j)|$ for $n\geq n_1$, then the $n$th coefficient of $f(q)\theta_\lambda(m\tau)$ is positive for $n\geq n_1+\frac{\lambda^2}{4}$.
\end{lemma}%

We are now ready to look at the positivity of the Fourier coefficients of the various functions.
Recall the definitions of $g_0$ and $g_1$ in (\ref{g0}) and (\ref{g1}), respectively.
We claim that all coefficients of
 $q^{-\frac14} g_0\eta^6$ and all but the constant coefficient of $-g_1\eta^6$ are positive.
Indeed, by Lemma~\ref{thetamult} the coefficients of $-E_2\theta_0$ and $-E_2\theta_1q^{-\frac14}$ are all
positive except the first giving the claim since
the coefficients of $\theta'_j$, $j=0,1$,
are positive and since
\begin{eqnarray*}
q^{-\frac14}g_0(\tau)\eta^6(\tau)&=&  10+48q+O \left( q^2\right),\\
g_1(\tau)\eta^6(\tau)&=& -1+70q+O \left(q^2 \right).
\end{eqnarray*}
Multiplying by $\eta^{-6}$, it is then clear that all coefficients of $g_0$ are positive.
For $n\geq 1$, the $n$th  coefficient of $-g_1\eta^6$ is  bounded by  $24\sigma_1(n)>6$.
Therefore, we have, using Lemma~\ref{multD}, that all coefficients with the exception of the
first of $q^{\frac14} g_1$ are negative.

Next  we  look at the positivity of the coefficients of $f_{0}$ and $f_{1}$.
Using the easily verified identity
\begin{equation}\label{thisq}
E_2E_4-E_6=3E_4' \, ,
\end{equation}
we obtain
$$
-11E_2(\tau)E_4(\tau)+5E_6(\tau)  =  -33E_4'(\tau)-6E_6(\tau)
 =  -6-4896 q-42768 q^2+72576 q^3+O\left(q^4\right) \, .
$$
\bea \nonumber
-11E_2(\tau)E_4(\tau)+5E_6(\tau) & = & -33E_4'(\tau)-6E_6(\tau) \,  \cr
& = & -6-4896 q-42768 q^2+72576 q^3+O\left(q^4\right) \, .
\eea
We denote the $n$th coefficient of this $q$-series by $a(n)$. It is given  for $n\geq 1$ by
\[
a(n)=-7920n\sigma_3(n)+3024\sigma_5(n)\geq 2529 n^5-7920n^4=:P_1(n),
\]
where we used that, for $n> 1$,
\[
\sigma_5(n)>n^5\qquad\qquad \sigma_3(n)\leq \frac{n^4}{16}+n^3.
\]
Note that  $P_1(n)>0$ for $n\geq 4$ and that
\[
a(n)-2\left(|a(0)|+|a(1)|+|a(2)|\right)>P_2(n):=P_1(n)-95340,
\]
which is for $n\geq 4$ strictly positive.
 By Lemma~\ref{thetamult}, it then follows that for $j=0,1$,  the $n$th Fourier coefficient of the function
\[
(5E_6-11E_2 E_4)\theta_j q^{-\frac{j^2}{4}}
\]
is   for $n\geq 4$ positive and bounded below by $P_2(n)$.
Moreover, since $E_4$ and  $\theta_j'q^{-j^2/4}$, $j=0, 1$, have positive Fourier coefficients, we obtain that for $n\geq 4$ the $n$th Fourier coefficients of $f_0q^{-1/4}$ and $f_1$ are positive and bounded below by $P_2(n)$.
Computing the first few coefficients gives
\begin{align*}
f_0(\tau) q^{-\frac14} &=120+21888q+200760q^2+1307520q^3+  O\left(q^4\right)  ,\\
f_1(\tau) &=-6-4380q+74160q^2+1127520q^3+O\left(q^4\right).
\end{align*}
We next aim to show that the Fourier coefficients of $\eta^{12}\myh_j$ are positive for all~$n$ except
for a finite number of possible exceptions.
For this purpose we first consider products of $f_0$ and $f_1$ with theta functions.
Since the coefficients of $f_0q^{-\frac14}$ are positive, it follows that for
$j=0,1$, the coefficients of $q^{(-j^2-1)/4}f_0\theta_j$ are all positive and thus also the coefficients of
$q^{(-j^2-1)/2} \theta_j(2 \tau)\theta_1(\tau) f_0(\tau)$ for $j=0, 1$.
Moreover, using again Lemma \ref{thetamult},  the $n$th coefficients of $q^{-j^2/4}f_1\theta_j$, $j=0, 1$,
are bounded below by
\[
P_3(n):=P_2(n)-2 \cdot (6+4380)
\]
which is positive for $n\geq 4$.
We determine the first coefficients as
\begin{align*}
f_1(\tau) \theta_0(\tau)&=-6-4392q+65400q^2+1275840q^3+  O\left(q^4\right)     ,\\
f_1(\tau) \theta_1(\tau) q^{-\frac14}&=-12-8760q+148308q^2+2246280q^3+  O\left(q^4\right).
\end{align*}
Thus, using that
$$
q^{-\frac14}\left(f_0(\tau)\theta_0(\tau)+f_1(\tau)\theta_1(\tau)\right)=   108+13368q+392844q^2+3955320q^3+O\left( q^4\right)
$$
gives that  all coefficients of
$(f_0\theta_0+f_1\theta_1)q^{-1/4}$ are positive and bounded below by $P_3(n)$. Therefore,
the same is true for $q^{-3/8}\myh_1 \eta^{12}$ and thus also for $q^{1/8} \myh_1$.

To treat $\myh_0$ and $\myh_2$, we apply Lemma \ref{thetamult} another time and
find that the $n$th coefficient of  $q^{-j^2/2}f_1(\tau)\theta_j(2 \tau)\theta_0(\tau)$
for $j=0,1$ is  bounded below by
\[
P_4(n):=P_3(n)-2\cdot(6+4392)
\]
which is positive for $n\geq 4$. Recall that  $q^{(-j^{2}-1)/2} \theta_{j}(2\t) \theta_{1}(\t) f_{0}(\t)$ has
positive coefficients. Thus  the coefficients of
$-q^{(-j^2-1)/2}\eta^{12}\myh_{2j}$, $j=0, 1$, are bounded below by~$P_4(n)$.
 We determine the first coefficients as
\begin{align*}
-q^{-\frac12}\eta^{12}(\tau)\xi_0(\tau)&=228+34992q+553040q^2+5298048q^3+   O\left(q^4\right)  ,\\
-\eta^{12}(\tau) \xi_2(\tau)&=-6-3912q+152940q^2+2070576 q^3+   O\left(q^4\right).
\end{align*}
Since all coefficients of $-q^{-1/2}\eta^{12}\xi_0$ are positive, the same is true for $-\xi_0$.
Moreover, the $n$th Fourier coefficient of $-\eta^{12} \xi_2$ is positive for $n\geq 2$ and bounded below  by $P_4(n)$ for $n\geq 4$.
By the proof of Lemma \ref{multD} we obtain that the coefficients of $-q^{\frac12}\xi_2$ are positive for $n\geq 2$ and bounded below by
$$
P_5(n):=P_4(n)-12 \cdot 3912
$$
which is positive for $n  \geq 4$.

The analysis of the three functions slightly differ from each other from now on and we start with the case $j=2$.
As shown above the $n$th coefficient of $-q^{1/2}\myh_2$ is bounded below by $P_5(n)$.  Note that all Fourier coefficients of   $\mathcal{F}_2$ are positive and that
$$
q^{\frac12}\left(- \myh_2(\tau)+4800 \mathcal{F}_2(\tau) \right) =  -6 - 1584q  +  115056q^2+3560256q^3
+O\left(q^4\right) .
$$
For   $n\geq 4$, the $n$th coefficient of this function are bounded below by $P_5(n)$, which
is bounded below by $24\cdot(6+1584)$. Therefore, we obtain, by  Lemma
\ref{multD} with $n_0=1$ and $k=24$ and by inspecting the first coefficient, that all coefficients with
positive exponent  of (\ref{xiD}) are positive.

In the case $j=0$, we proceed similarly. Note that for $n\geq 1$, the $n$th Fourier coefficient of
$\mathcal{F}_0$ is positive. Thus for $n\geq 1$ the $n$th Fourier coefficients of  $- \myh_0+4800 \mathcal{F}_0$ is positive and
$$
- \myh_0(\tau)+4800 \mathcal{F}_0(\tau)  =   -972+42528 q+985464q^2+1497196q^3+ O\left(q^4\right).
$$
From the above analysis, we moreover obtain that for   $n\geq 4$, the $n$th coefficient of this function is bounded below by  $P_5(n)$, which is bounded below by $24\cdot 972$.
Therefore we obtain by  Lemma~\ref{multD} with $n_0=0$ and $k=24$, that all positive coefficients of (\ref{xiD}) are positive.

We finally consider the case $j=1$. We have to compare the   coefficients of $\myh_1$ with the associated contribution coming from class numbers. Recall that for $n\geq 4$ the $n$th coefficient of $q^{1/8}\myh_1$ is bounded by $P_3(n)$. Next note  that
$$
q^{\frac18}\left(\myh_1(\tau)+4800 \mathcal{F}_1(\tau) \right)= 108 + 9972q + 568044q^2 + 10477416q^3 +O \left(q^4 \right) \, .
$$
Since
$$
\mathcal{F}_1(\tau) = \sum_{\ell\geq 1}H(8\ell-1)q^{\ell-\frac18} \, ,
$$
it is enough by the above considerations to show that  for $n \geq 4$
$$
P_3(n)>4800H(8n-1).
$$
By (\ref{classtrivial}) it is enough to show that
$$
P_3(n)>4800(8n-1)
$$
which is indeed satisfied for $n \geq 4$.

\newpage

\appendix
\section{Black hole degeneracies for $m=1,2,3,4$}

For the black holes in string theory on $K3 \times T^{2}$, the degeneracies are a function of the T-duality
invariants~$(M^{2}/2, N^{2}/2, M \dot N) = (m,n,\ell)$. As explained in~\S\ref{1/44d}, they are the
Fourier coefficient~$c(n,\ell)$ of the mock Jacobi form~$\psi_{m}$ of index~$m$ for $n \ge m$.
By the elliptic invariance, it is enough to
consider~$\ell = 0,\dots, m+1$. We list the first few coefficients of the mock Jacobi forms $\psi_{m}$
for the first four positive values of~$m$.

\bigskip

{$\bf m=1$}

\bigskip

\begin{tabular}{|c|cccccccc|cccccc}
\hline
$n$ & $-1$ & 0 & 1 & 2 & 3 & 4 & 5 & 6    \\
\hline
$\ell=0$ &  $-48$ &  $648$ &  $50064$ &  $1127472$ &  $16491600$ &  $185738352$ &  $1737283968$ &  $14086119024$ \\
$\ell=1$ &  $3$ &  $600$ &  $25353$ &  $561576$ &  $8533821$ &  $100390104$ &  $977183520$ &  $8203464720$ \\
\hline
\end{tabular}

\bigskip

\begin{tabular}{|c|cccc|cccccccccc}
\hline
$n$ & 7 & 8 & 9 &10  \\
\hline
$\ell=0$ & $101777516400$ & $668043042720$ & $4040083875024$ & $22756537895040$ \\
$\ell=1$ & $61077837780$ & $411421124040$ & $2544746970243$ & $14618739930912$ \\
\hline
\end{tabular}

\bigskip

\bigskip

{$\bf m=2$}

\bigskip

\begin{tabular}{|c|ccccccc|ccccccc}
\hline
$n$ & $-1$ & 0 & 1 & 2 & 3 & 4 & 5     \\
\hline
$\ell=0$ &  $-648$ &  $12800$ &  $1127472$ &  $32861184$ &  $632078672$ &  $9337042944$ &  $113477152800$\\
$\ell=1$ &  $72$ &  $8376$ &  $561576$ &  $18458000$ &  $392427528$ &  $6216536784$ &  $79330416536$\\
$\ell=2$ &  $-4$ &  $-1152$ &  $50064$ &  $3859456$ &  $110910300$ &  $2073849984$ &  $29495727056$\\
\hline
\end{tabular}

\bigskip

\begin{tabular}{|c|cccc|cccccccccc}
\hline
$n$ & 6 & 7 & 8 & 9  \\
\hline
$\ell=0$ &  $1181763743744$ &  $10838236934808$ &  $89288280271872$ &  $670746948265232$ \\
$\ell=1$ &  $855667882536$ &  $8055449338200$ &  $67714250601728$ &  $516898213691112$ \\
$\ell=2$ &  $343972015104$ &  $3437700768840$ &  $30312295881600$ &  $240704209521024$ \\
\hline
\end{tabular}

\bigskip

\bigskip

{$\bf m=3$}

\bigskip

\begin{tabular}{|c|cccccc|cccccccc}
\hline
$n$ & $-1$ & 0 & 1 & 2 & 3 & 4     \\
\hline
$\ell=0$ &  $-6404$ &  $153900$ &  $16491600$ &  $632078672$ &  $16193130552$ &  $315614079072$\\
$\ell=1$ &  $972$ &  $85176$ &  $8533821$ &  $392427528$ &  $11232685725$ &  $233641003920$\\
$\ell=2$ &  $-96$ &  $-15600$ &  $1127472$ &  $110910300$ &  $4173501828$ &  $100673013264$\\
$\ell=3$ &  $5$ &  $1728$ &  $130329$ &  $18458000$ &  $920577636$ &  $26563753008$\\
\hline
\end{tabular}

\bigskip

\begin{tabular}{|c|cccc|cccccccccc}
\hline
$n$ & 5 & 6 & 7 & 8   \\
\hline
$\ell=0$ &  $4980146121600$ &  $66223829146464$ &  $763810107420924$ &  $7808500872944344$\\
$\ell=1$ &  $3838665438606$ &  $52438270948872$ &  $616509025474839$ &  $6394025215102200$\\
$\ell=2$ &  $1817641213584$ &  $26523447693936$ &  $327561687731700$ &  $3530513346970608$\\
$\ell=3$ &  $543037538313$ &  $8689043006928$ &  $115301073750300$ &  $1317086884043616$\\
\hline
\end{tabular}

\bigskip

\newpage

{$\bf m=4$}

\bigskip

\begin{tabular}{|c|cccccc|cccccccc}
\hline
$n$ & $-1$ & 0 & 1 & 2 & 3 & 4      \\
\hline
$\ell=0$ &  $-51396$ &  $1410048$ &  $185738352$ &  $9337042944$ &  $315614079072$ &  $7999169992704$\\
$\ell=1$ &  $9600$ &  $700776$ &  $100390104$ &  $6216536784$ &  $233641003920$ &  $6264458136216$\\
$\ell=2$ &  $-1296$ &  $-154752$ &  $16491600$ &  $2073849984$ &  $100673013264$ &  $3093523125120$\\
$\ell=3$ &  $120$ &  $23328$ &  $1598376$ &  $392427528$ &  $26563753008$ &  $987647838816$\\
$\ell=4$ &  $-6$ &  $-2304$ &  $-209304$ &  $32861184$ &  $4173501828$ &  $203003283456$\\
\hline
\end{tabular}

\bigskip

\begin{tabular}{|c|cccc|cccccccccc}
\hline
$n$ & 5 & 6 & 7 & 8   \\
\hline
$\ell=0$ &  $161166049715136$ &  $2690630398144512$ &  $38396325233501604$ &  $479643192755712000$\\
$\ell=1$ &  $130483874926824$ &  $2226273321514872$ &  $32263019501551200$ &  $407734088790024888$\\
$\ell=2$ &  $70154254155648$ &  $1268909447328000$ &  $19194759843735744$ &  $250750639230059136$\\
$\ell=3$ &  $25364019402816$ &  $501130864684008$ &  $8102401641823224$ &  $111544858411221936$\\
$\ell=4$ &  $6153448819056$ &  $136676238618624$ &  $2415415078450044$ &  $35695523741819136$\\
\hline
\end{tabular}

\bigskip

\begin{tabular}{|c|ccc|ccccccccccc}
\hline
$n$ & 9 & 10 & 11    \\
\hline
$\ell=0$ &  $5343131141125608240$ &  $53865362293195763712$ &  $497287540606193791776$\\
$\ell=1$ &  $4583919031715817912$ &  $46559483089512998904$ &  $432557670343025950296$\\
$\ell=2$ &  $2895018873817853040$ &  $30059106955693337088$ &  $284514174765163372992$\\
$\ell=3$ &  $1344064501276102440$ &  $14463582110776040904$ &  $141148056895219254264$\\
$\ell=4$ &  $455984813319184992$ &  $5155065821726530560$ &  $52491288465592800984$\\
\hline
\end{tabular}

\bigskip

\bibliographystyle{amsplain}	

\providecommand{\bysame}{\leavevmode\hbox to3em{\hrulefill}\thinspace}
\providecommand{\MR}{\relax\ifhmode\unskip\space\fi MR }
\providecommand{\MRhref}[2]{%
  \href{http://www.ams.org/mathscinet-getitem?mr=#1}{#2}
}
\providecommand{\href}[2]{#2}

\end{document}